\documentclass[11pt,a4paper]{article}
\usepackage[T1]{fontenc}
\usepackage[utf8]{inputenc}
\usepackage{lmodern}
\usepackage[margin=2.65cm]{geometry}
\usepackage{amsmath,amssymb,amsthm,mathrsfs}
\usepackage{microtype}
\usepackage{xcolor}
\usepackage[colorlinks=true,linkcolor=blue,citecolor=blue,urlcolor=blue]{hyperref}
\hypersetup{pdftitle={Projective geometry of second-order supersymmetry},pdfauthor={Stevan Cebrian and Mikhail S. Plyushchay}}
\numberwithin{equation}{section}
\newtheorem{proposition}{Proposition}[section]
\newcommand{\Bol}{\operatorname{Bol}}

\title{Projective geometry of second-order supersymmetry}
\author{
Stevan Cebrian and Mikhail S. Plyushchay\\[3pt]
{\small\itshape Departamento de F\'{\i}sica, Universidad de Santiago de Chile,}\\
{\small\itshape Avenida V\'ictor Jara 3493, Santiago, Chile}\\[3pt]
{\small\texttt{\textcolor{blue}{stevan.cebrian@usach.cl,\, mikhail.plyushchay@usach.cl}}}
}
\date{}
\begin{document}
\maketitle
\begin{abstract}

We show that anomaly-free second-order supersymmetry (SUSY) in quantum mechanics admits a projective-geometric formulation in terms of a projective connection and a vector field preserving it. Once the physical coordinate and energy origin are fixed, these data determine the partner Hamiltonians and supercharges;  the field's quadratic invariant fixes the central parameter and candidate singlet energies. 
The fixed-point structure of this field yields explicit criteria for
smooth continuation of the quantum operators through superpotential zeros. 
The classically fictitious similarity transformation
selects affine data whose covariant factors generate the Schwarzian
quantum correction.  
Dimensional reduction of a Riemannian model reproduces these affine 
data through the volume measure, while a coordinate adapted to the
selected field implements coupling-constant metamorphosis. 
A one-gap Lam\'e pair illustrates the relation between the SUSY auxiliary
geometry and the Gelfand--Dikii--Virasoro structure governing physical
solution products. 
Its Lax--Novikov integral acts on Schr\"odinger solutions through
the first-order action of a projective stabilizer;
  in the reflectionless limit it
becomes proportional, in the physical coordinate, to the auxiliary
third Bol operator.

\end{abstract}

\section{Introduction}
\label{sec:intro}
Projective geometry enters physics whenever a choice of coordinate is
defined only up to a fractional-linear transformation. In conformal
mechanics, the \(\mathfrak{sl}(2,\mathbb R)\) generators act in this
way on time, and a change of evolution parameter relates free and
oscillator dynamics \cite{AFF1976,Niederer1973}. The same symmetry
underlies the early superconformal extensions
\cite{AkulovPashnev1983,FubiniRabinovici1984}. In quantum field theory
and gravity, conformal symmetry and reparametrizations play a central
role \cite{BPZ1984,Polyakov1981}. Their projective content is particularly
visible in the Schwarzian dynamics of nearly-AdS$_2$
Jackiw--Teitelboim gravity and the low-energy Sachdev--Ye--Kitaev model
\cite{MaldacenaStanfordYang2016,MaldacenaStanford2016,Mertens2018};
see refs.~\cite{Rosenhaus2019,MertensTuriaci2023} for reviews.
In this paper, the relevant projective coordinate is spatial: it is
the ratio of two independent solutions of an auxiliary second-order
equation.

Supersymmetric quantum mechanics realizes supersymmetry (SUSY) by
relating partner Schr\"odinger Hamiltonians through differential intertwiners
\cite{CooperKhareSukhatme1995,Junker2019}. Higher-order supercharges
can close on a polynomial in the Hamiltonian
\cite{AndrianovIoffeSpiridonov1993,Plyushchay2000}.
The classical construction of this nonlinear SUSY, its quantum anomaly,
and its relation to quasi-exact solvability were developed in
refs.~\cite{KlishevichPlyushchay2001,Plyushchay2004}.
Second order is exceptional: an anomaly-free quantization exists for
an arbitrary local nonvanishing superpotential, with a real central
parameter controlling the possible singlet energies
\cite{KlishevichPlyushchay2001}. Subsequent polynomial-kernel
descriptions exhibit fractional-linear covariance and discriminant
invariants \cite{Tanaka2003,BagchiTanaka2009}.
The Schwarzian quantum correction and the coupling-constant
metamorphosis relating second- and first-order SUSY were identified in
ref.~\cite{Plyushchay2017}. The classical factorization behind that
correction is also central to the position-dependent-mass analysis
of ref.~\cite{BravoPlyushchay2016}.

We identify the geometric data that reconstruct all local differential
operators of this established SUSY family: a projective connection and
one vector field preserving it, with the physical coordinate and energy
origin fixed. The connection determines the auxiliary charge equation; 
the selected field determines the multiplicative factors relating
auxiliary solutions to charge zero modes (the dressing), the difference
between the partner potentials, and the Hamiltonian action on the charge
kernel. Its quadratic invariant fixes the central parameter \(C\).
 Classical solution-product and
Ermakov--Pinney identities underlie this reconstruction
\cite{Pinney1950,Brezhnev2008Integrability}. The resulting correspondence
also gives a regular parametrization at superpotential zeros, where
the usual formulas contain separately divergent terms. Physical
realizations additionally require compatible domains and boundary
conditions.

The geometric formulation has three consequences. It interprets the
classically fictitious factorization as the choice of an affine
connection whose covariant factors produce the Schwarzian correction.

It identifies the candidate singlet energies with the selected field's
conjugacy invariant and gives necessary and sufficient cancellation
conditions at simple and double superpotential zeros.

Finally, it distinguishes auxiliary projective symmetry from physical
integrability. The differential operators expressing projective
covariance are the Bol operators \cite{Bol1949,GustafssonPeetre1989}.
The third Bol operator annihilates the selected field, but it need
not commute with either physical Hamiltonian. The Lam\'e pair and
its reflectionless limit make this distinction explicit, using their
known coexisting SUSYs and Lax--Novikov integrals
\cite{CorreaEtAl2008,CorreaJakubskyPlyushchay2008,PlyushchayArancibiaNieto2011}.
Higher-order bosonic integrals also occur in nonlinear superconformal
mechanics \cite{LeivaPlyushchay2003}, where physical realization
requires care with domains \cite{CorreaDelOlmoPlyushchay2005}.

The comparison uses the distinction between the Gelfand--Dikii (GD)
second Hamiltonian structure of the Korteweg--de Vries (KdV) equation
and its Lax representation \cite{Zuber1993}, together with the
solution-product description of finite-gap integrability
\cite{Brezhnev2008Integrability}.
We relate these physical spectral data to the additional auxiliary
data selected by SUSY. The reflectionless limit gives an explicit
coincidence of the auxiliary Bol expression with the commuting Lax
integral in the physical coordinate.

Section~\ref{sec:classical} fixes the classical system and its
anomaly-free quantization. Section~\ref{sec:factorization} derives
the affine and projective content of the fictitious factorization.
Section~\ref{sec:reconstruction} gives the reconstruction, derives
its spectral and regularity consequences, and explains the distinction
from higher orders. Section~\ref{sec:curved} recovers the connection
by curved reduction and identifies the geometric map behind
metamorphosis. 
Section~\ref{sec:integrable} relates the Lam\'e Lax integral to
the Gelfand--Dikii projective stabilizer and the SUSY auxiliary connection.
Section~\ref{sec:discussion} discusses the results and open questions;
the appendix supplies the operator checks.

We use unit mass, \(\hbar>0\), and \(D_x=d/dx\). The real
superpotential \(W\) is smooth; \(a\) is an energy origin and
\(C\) is real. Square-root formulas are first written on a
\(W>0\) interval; logarithmic derivatives extend to fixed-sign
intervals using \(\sqrt{|W|}\). Differential identities use formal
adjoints in \(L^2(dx)\); physical statements additionally specify
admissible states and compatible domains.

\section{Classical nonlinear SUSY and its quantum correction}
\label{sec:classical}
Let \(\{x,p\}=1\) and \(\{\theta^+,\theta^-\}=-i\), where
the conjugate Grassmann variables obey
\(N=\theta^+\theta^-\), \(N^2=0\). For integer \(n\geq1\),
\begin{equation}
 H_n=\tfrac12(p^2+W^2)+nW'N+a,\qquad
 Q_n^+=2^{-n/2}(W+ip)^n\theta^+,\qquad Q_n^-=(Q_n^+)^*
 \label{eq:classical-holo}
\end{equation}
obey \(\{H_n,Q_n^\pm\}=0\) and
\(\{Q_n^+,Q_n^-\}=-i(H_n-a)^n\)
\cite{KlishevichPlyushchay2001,Plyushchay2004}.
The integer boson--fermion coupling cancels the phase evolution of
\((W+ip)^n\) against that of \(\theta^+\); the remaining terms
vanish by \(N\theta^+=0\). This gives nonlinear SUSY already
at the classical level, without restricting \(W\).

At second order, an additional real constant \(C\) is allowed,
\mbox{yielding~\cite{KlishevichPlyushchay2001,Plyushchay2017}}
\begin{align}
 H_{\rm cl}&=\tfrac12(p^2+W^2-C/W^2)+2W'N+a,\nonumber\\
 Q_{\rm cl}^+&=\tfrac12[(W+ip)^2+C/W^2]\theta^+,\qquad
 \{Q_{\rm cl}^+,Q_{\rm cl}^-\}=-i[(H_{\rm cl}-a)^2+C].
 \label{eq:classical-pair}
\end{align}
The combination \(W^2-C/W^2\) suggests a conformal origin:
for \(W=\omega x\), its bosonic Hamiltonian is
\(\tfrac12(p^2+\omega^2x^2+g/x^2)+a\), with
\(g=-C/\omega^2\), the oscillator form of de Alfaro--Fubini--Furlan
(AFF) conformal mechanics
\cite{AFF1976}. For general \(W\), however, replacing \(x\)
by \(W(x)\) changes the kinetic term. The universal geometric
statement will concern an auxiliary projective symmetry.

We use fermionic Weyl ordering:
\(\widehat N=\tfrac12[\widehat\theta^+,\widehat\theta^-]
=\hbar\sigma_3/2\), with
\(\widehat\theta^\pm=\sqrt{\hbar}(\sigma_1\pm i\sigma_2)/2\).
Quantizing \(p\) as \(-i\hbar D_x\) and \((W+ip)^2\) as
the operator square \((\hbar D_x+W)^2\) then gives 
\begin{equation}
 2(H_\pm^{(0)}-a)=-\hbar^2D_x^2+W^2\pm2\hbar W'-C/W^2,
 \qquad 2q_+^{(0)}=(\hbar D_x+W)^2+C/W^2.
\end{equation}
The required relation \(q_+H_-=H_+q_+\) fails by
\begin{equation}
 q_+^{(0)}H_-^{(0)}-H_+^{(0)}q_+^{(0)}=-\frac{\hbar^3}{4}W'''.
 \label{eq:anomaly}
\end{equation}
This quantum anomaly is distinct from SUSY breaking, which concerns
charge zero modes after the algebra has been established.

Keeping \(W\), the energy origin \(a\), and the coefficients
of \(D_x^2\) and \(D_x\) in the charge fixed, comparison of
derivative coefficients in the intertwining relation requires opposite
scalar corrections,
 \(H_\pm=H_\pm^{(0)}+\delta V/2\) and
\(q_+=q_+^{(0)}-\delta V/2\). The remaining condition is
\(W\delta V'+2W'\delta V=\hbar^2W'''/2\).
Absorbing its integration constant in \(C\) gives
\begin{equation}
 \Delta=\hbar^2\left(\frac{W''}{2W}-\frac{W'^2}{4W^2}\right),
 \qquad U_C=\Delta-C/W^2.
 \label{eq:Delta}
\end{equation}
The resulting established family is
\begin{equation}
 2(H_\pm-a)=-\hbar^2D_x^2+W^2\pm2\hbar W'+U_C,
 \qquad 2q_+=(\hbar D_x+W)^2-U_C,\qquad q_-=q_+^\dagger.
 \label{eq:operators}
\end{equation}
With the matrix operators and grading
\[
 \mathcal H=\begin{pmatrix}H_+&0\\0&H_-\end{pmatrix},\qquad
 \mathcal Q_+=\begin{pmatrix}0&q_+\\0&0\end{pmatrix},\qquad
 \mathcal Q_-=\mathcal Q_+^\dagger,\qquad \Gamma=\sigma_3,
\]
one obtains
\begin{equation}
 [\mathcal H,\mathcal Q_\pm]=0,\qquad \mathcal Q_\pm^2=0,
 \qquad \{\mathcal Q_+,\mathcal Q_-\}=(\mathcal H-a)^2+C\mathbf1.
 \label{eq:superalgebra}
\end{equation}
Appendix~\ref{app:proofs} justifies the correction prescription and
the quadratic closure. We now turn to the geometric information
encoded in \(U_C\).

\section{From classical factorization to projective geometry}
\label{sec:factorization}
\subsection{The fictitious similarity transformation}
The correction in \eqref{eq:Delta} can be generated by a prescription
already formulated at the classical level. For a positive function
\(\zeta(x)\), ordinary multiplication gives
\begin{equation}
 p^2=\zeta p\zeta^{-2}p\zeta
     =(-i\zeta p\zeta^{-1})(i\zeta^{-1}p\zeta).
 \label{eq:classical-insertion}
\end{equation}
This is the \emph{fictitious similarity transformation} of
refs.~\cite{Plyushchay2017,BravoPlyushchay2016}: it leaves the classical
observable and canonical variables unchanged. Retaining the factor order
in quantization produces
\begin{align}
 \mathsf A_\zeta&=\hbar\zeta^{-1}D_x\zeta
                =\hbar(D_x+\gamma_\zeta),\qquad
                  \gamma_\zeta=(\log\zeta)',\nonumber\\
 K_\zeta&=\mathsf A_\zeta^\dagger\mathsf A_\zeta
       =-\hbar^2D_x^2+\hbar^2(\gamma_\zeta^2-\gamma_\zeta').
 \label{eq:Kzeta}
\end{align}
Thus the classically invisible function becomes quantum factorization
data. This is not a similarity of the already quantized free Hamiltonian;
such a common similarity of a Hamiltonian and a charge could not change
whether their commutator vanishes.

For second-order SUSY, choose \(\zeta=W^{-1/2}\). Then
\(K_\zeta=-\hbar^2D_x^2+\Delta\). The same insertion in the
classical charge, with \(Z=W+ip\),
\(Z^2=(\zeta Z\zeta^{-1})(\zeta^{-1}Z\zeta)\),
gives
\begin{equation}
 \left(\hbar D_x+W+\frac{\hbar W'}{2W}\right)
 \left(\hbar D_x+W-\frac{\hbar W'}{2W}\right)
 = (\hbar D_x+W)^2-\Delta.
 \label{eq:charge-factors}
\end{equation}
One prescription therefore supplies \(+\Delta/2\) in both
Hamiltonians and \(-\Delta/2\) in the charge: these are precisely
the opposite corrections required in section~\ref{sec:classical},
at fixed \(C\). The classical \(C/W^2\) term is unaffected. We now identify the geometric information
encoded by these factors.

\subsection{Affine connections and the second Bol operator} 
A basis of local solutions of \(K_\zeta v=0\) is 
\begin{equation}
 v_0=\zeta^{-1},\qquad v_1=\zeta^{-1}\xi_\zeta,
 \qquad \xi_\zeta'=\zeta^2.
\end{equation}
A change of solution basis transforms their ratio \(\xi_\zeta\)
by a fractional-linear map. Such a ratio is a \emph{developing
coordinate}: it represents the local projective structure by a
coordinate on the projective line, with another ratio chart available
where the denominator vanishes. 
In the fixed coordinate \(x\), its basis-independent coefficient
is the Schwarzian derivative, 
\begin{equation}
 \{\xi;x\}=\frac{\xi'''}{\xi'}-\frac32\left(\frac{\xi''}{\xi'}\right)^2,
 \qquad \{(\alpha\xi+\beta)/(\gamma\xi+\delta);x\}=\{\xi;x\}.
 \label{eq:Schwarzian}
\end{equation}
The determinant of the fractional-linear map is nonzero.

For a coordinate change \(x=g(t)\), \(g'>0\), a density of weight
\(\rho\) has coefficients \(v_t=(g')^\rho v_x\circ g\), where
the subscripts label coordinate descriptions, not derivatives. 
Assigning weight \(+1/2\) to \(\zeta\) makes
\(\xi_\zeta'=\zeta^2\) coordinate-compatible. Define
\begin{equation}
 \Gamma_\zeta=2\gamma_\zeta=\frac{\xi_\zeta''}{\xi_\zeta'},\qquad
 \mathcal U_\zeta=\Gamma_\zeta'-\frac12\Gamma_\zeta^2
                 =\{\xi_\zeta;x\}.
 \label{eq:affine}
\end{equation}
They transform as
\begin{equation}
 \Gamma_{\zeta,t}=g'\Gamma_\zeta\circ g+\frac{g''}{g'},\qquad
 \mathcal U_{\zeta,t}=(g')^2\mathcal U_\zeta\circ g+\{g;t\}.
 \label{eq:connection-laws}
\end{equation}
The first law defines an affine connection and the second a projective
connection. Their relation in \eqref{eq:affine} is a Riccati, or
Miura-type, relation. The covariant derivative on weight \(\rho\) 
is \(D_x-\rho\Gamma_\zeta\). In the fixed coordinate \(x\),
the corresponding differential expressions satisfy 
\begin{equation}
 \left(D_x-\frac12\Gamma_\zeta\right)
 \left(D_x+\frac12\Gamma_\zeta\right)
 =D_x^2+\frac12\mathcal U_\zeta,
 \qquad K_\zeta=-\hbar^2\Bol_2^{\mathcal U_\zeta}.
 \label{eq:Bol2}
\end{equation}
Here \(\Bol_2^{\mathcal U}=D_x^2+\mathcal U/2\) is the second Bol
operator. More generally, the \(n\)th Bol operator maps densities
of weight \((1-n)/2\) to weight \((1+n)/2\); in a projective
coordinate, where the projective coefficient vanishes, it is the
ordinary \(n\)th derivative \cite{Bol1949,GustafssonPeetre1989}. 
Intrinsically, \(\Bol_2^{\mathcal U}\) maps weight \(-1/2\)
to weight \(3/2\); its identification with the scalar expression
\(-K_\zeta/\hbar^2\) in \eqref{eq:Bol2} uses their coefficients
in the chosen coordinate. 
The covariance of this map follows from
\eqref{eq:connection-laws} \cite{OvsienkoTabachnikov2005,Ovsienko2006}.

The fictitious insertion thus selects an affine connection and a
\emph{solution line}, the one-dimensional subspace
\(\mathbb R\zeta^{-1}\) of the local kernel. The associated
covariant factors produce the projective operator.
 Different solution lines can give
the same second-order operator: if \(\widehat\xi\) is a M\"obius
transform of \(\xi_\zeta\) and \(\widehat\zeta^2=\widehat\xi'\)
on an orientation-preserving chart, then
\(K_{\widehat\zeta}=K_\zeta\). A general change of \(\zeta\)
instead changes the quantization. The classical identity alone does
not fix this quantum information. The choice \(\zeta=W^{-1/2}\)
reproduces the correction required by SUSY, but is not a unique
factorization of the resulting operator.

\subsection{The complete auxiliary connection} 
For this representative, on a \(W>0\) interval write
\(R=\sqrt W\), \(\chi'=W\), and 
\begin{equation}
 z'=W^{-1},\qquad \Delta=-\frac{\hbar^2}{2}\{z;x\},\qquad
 \mathcal L_C=-\hbar^2D_x^2+U_C,
 \qquad \mathcal U_C=-\frac{2U_C}{\hbar^2}.
 \label{eq:auxiliary}
\end{equation}
Since \(\chi'=W\),
\(e^{\chi/\hbar}(\hbar D_x+W)e^{-\chi/\hbar}=\hbar D_x\).
This conjugation removes \(W\) from the first-order factor and hence
the first-derivative term from the second-order charge. Equivalently,
\begin{equation}
 2q_+=-e^{-\chi/\hbar}\mathcal L_Ce^{\chi/\hbar},\qquad
 2q_-=-e^{\chi/\hbar}\mathcal L_Ce^{-\chi/\hbar}.
 \label{eq:charge-gauge}
\end{equation}
The auxiliary equation \(\mathcal L_C\psi=0\) is in Liouville
normal form, with constant leading coefficient and no first derivative.
Its potential \(U_C\) is the common correction in \eqref{eq:operators};
the complete common potential is \(W^2+U_C\).

Let \(j=z'>0\). If \(\phi\) solves the reference equation
\((-\hbar^2D_z^2+V_{\rm ref})\phi=0\), then
\(\psi(x)=j^{-1/2}\phi(z(x))\) solves
\((-\hbar^2D_x^2+V(x))\psi=0\), with 
\begin{equation}
 V(x)=j^2V_{\rm ref}(z(x))-\frac{\hbar^2}{2}\{z;x\}.
 \label{eq:Liouville}
\end{equation}
Taking \(V_{\rm ref}=-C\) yields
\begin{equation}
 \mathcal L_C=j^{3/2}(-\hbar^2D_z^2-C)j^{1/2},
 \qquad \psi=R\phi(z).
 \label{eq:transport}
\end{equation}
Here pullback by \(z=z(x)\) is understood in the mixed-coordinate
operator expression. The two contributions to \(U_C\) are the transported
constant and the Schwarzian of the coordinate change. This is a map
of stationary equations between density weights, not a unitary
equivalence of physical Hamiltonians, as discussed in
section~\ref{sec:stationary-metamorphosis}.

For two independent reference solutions, put
\(f_{\rm ref}=\phi_1/\phi_2\). The full auxiliary solution ratio is
\(f(x)=f_{\rm ref}(z(x))\). The Schwarzian chain rule gives
\begin{equation}
 \mathcal U_C=\{f;x\}=\frac{2C}{\hbar^2W^2}+\{z;x\},
 \qquad U_C=-\frac{\hbar^2}{2}\{f;x\}.
 \label{eq:full-ratio}
\end{equation}
One may take \(f_{\rm ref}=z\) for \(C=0\),
\(e^{2\beta z/\hbar}\) for \(C=-\beta^2<0\), \(\beta>0\), or
\(\tan(\sqrt C\,z/\hbar)\) for \(C>0\). 
Thus \(z\) and \(f\) serve different purposes: \(z\) makes
the selected vector field a translation, whereas \(f\) is a
developing coordinate of the complete auxiliary equation. They can
coincide for \(C=0\). 
The Schwarzian coefficient in \eqref{eq:Liouville} is fixed by the
density weights. The value of \(C\) is additional SUSY data, not a
constant fixed by projective covariance alone.

\section{Projective reconstruction and its consequences}
\label{sec:reconstruction}

\subsection{The selected projective generator}
Let \(\psi_1,\psi_2\) be a real fundamental pair for
\(\psi''+\frac{1}{2}\mathcal U\psi=0\), with constant nonzero Wronskian
\(\varpi=\psi_1\psi_2'-\psi_1'\psi_2\). Their products span
\begin{equation}
 \ker_{\rm loc}\Bol_3^{\mathcal U}
   =\operatorname{span}_{\mathbb R}\{\psi_1^2,\psi_1\psi_2,\psi_2^2\},
 \qquad \Bol_3^{\mathcal U}=D_x^3+2\mathcal U D_x+\mathcal U'.
 \label{eq:Bol3}
\end{equation}
This is the third Bol operator. The product space is the
\emph{symmetric square} of the two-dimensional solution space.
The products have density weight \(-1\), appropriate to coefficients
of vector fields. With \(f=\psi_1/\psi_2\),
\(f'=-\varpi/\psi_2^2\), the fields
\begin{equation}
 -\frac{\psi_2^2}{\varpi}\partial_x=\partial_f,\qquad
 -\frac{\psi_1\psi_2}{\varpi}\partial_x=f\partial_f,\qquad
 -\frac{\psi_1^2}{\varpi}\partial_x=f^2\partial_f
 \label{eq:projective-fields}
\end{equation}
generate \(\mathfrak{sl}(2,\mathbb R)\).
For \(g_\epsilon(x)=x+\epsilon\xi(x)\), the infinitesimal form of
\eqref{eq:connection-laws} is
\begin{equation}
 \delta_\xi\mathcal U=\xi\mathcal U'+2\xi'\mathcal U+\xi'''
                     =\Bol_3^{\mathcal U}\xi.
 \label{eq:variation}
\end{equation}
An element of this kernel is a \emph{projective stabilizer}: its
vector field preserves the connection. This local symmetry need not
give a conserved conformal charge of the physical Hamiltonian.

Specialize now to \(\mathcal U=\mathcal U_C\), and use
\(\psi_1,\psi_2\) for a fundamental pair of this auxiliary equation.
Equation~\eqref{eq:Delta} implies 
\begin{equation}
 I[W,\mathcal U_C]:=2WW''-W'^2+2\mathcal U_CW^2
                   =\frac{4C}{\hbar^2}.
 \label{eq:invariant}
\end{equation} 
Since \(I'=2W\Bol_3^{\mathcal U_C}W\), constancy of \(I\)
implies the stabilizer equation on every interval where \(W\ne0\).
Conversely, a stabilizer makes \(I\) constant. 
Equation~\eqref{eq:Bol3} therefore gives
\begin{equation}
 W=M_{11}\psi_1^2+2M_{12}\psi_1\psi_2+M_{22}\psi_2^2,
 \qquad C=\hbar^2\det(\mathsf M)\varpi^2,
 \label{eq:quadratic-form}
\end{equation}
where \(\mathsf M\) is a constant symmetric matrix; substitution in
\eqref{eq:invariant} proves the determinant formula. 
For \(W=R^2>0\), the same identity reads
\(R''+\mathcal U_CR/2=C/(\hbar^2R^3)\), the Ermakov--Pinney
equation. Its positive solution \(R\) is the square root of the
quadratic form in \eqref{eq:quadratic-form}
\cite{Pinney1950}.

The invariant \(C\) characterizes the connection together with the
selected stabilizer \(W\partial_x\). A change of auxiliary solution
basis changes \(\mathsf M\) and the Wronskian \(\varpi\), but leaves
the combination \(\det(\mathsf M)\varpi^2\), and hence \(C\),
unchanged. By contrast, choosing a different stabilizer at fixed
connection generally changes both the physical system and \(C\).

In the developing coordinate \(f=\psi_1/\psi_2\), the selected
vector field becomes
\begin{equation}
 W\partial_x=P(f)\partial_f,\qquad
 P(f)=-\varpi\bigl(M_{22}+2M_{12}f+M_{11}f^2\bigr).
\end{equation}
The quadratic polynomial \(P\) generates a flow of M\"obius
transformations. To identify its matrix generator, write
\[
 f_s=\frac{a_s f+b_s}{c_s f+d_s},\qquad
 \begin{pmatrix}a_s&b_s\\c_s&d_s\end{pmatrix}
 =e^{s\mathsf T_W}\in SL(2,\mathbb R),
 \qquad f_0=f,
\]
where \(s\) is the flow parameter and \(\mathsf T_W\) is a constant
traceless matrix. Differentiating at \(s=0\) gives
\[
 \left.\frac{df_s}{ds}\right|_{s=0}
 =(\mathsf T_W)_{12}
   +2(\mathsf T_W)_{11}f
   -(\mathsf T_W)_{21}f^2.
\]
Matching this expression with \(P(f)\) determines the generator
directly in terms of the quadratic-form coefficients:
\begin{equation}
 \mathsf T_W
 =\varpi
 \begin{pmatrix}
  -M_{12}&-M_{22}\\
   M_{11}& M_{12}
 \end{pmatrix},
 \qquad
 \det\mathsf T_W=\frac{C}{\hbar^2},
 \qquad
 \mathsf T_W^2=-\frac{C}{\hbar^2}\mathbf1.
 \label{eq:matrix-invariant}
\end{equation}
Thus \(\mathsf M\) expresses \(W\) as a quadratic form in the
auxiliary solutions, while \(\mathsf T_W\) represents the same
vector field as a generator of projective transformations.
The determinant identity follows from
\(\det\mathsf T_W=\varpi^2\det\mathsf M\); the identity for
\(\mathsf T_W^2\) follows from its vanishing trace.

For a nonzero selected field, \(C<0\), \(C=0\), and \(C>0\)
therefore correspond to hyperbolic, parabolic, and elliptic
projective flows, respectively. These have two distinct real
fixed points, one double real fixed point, or no real fixed
point on the projective line. A fixed point may lie at infinity.
Here \(C/\hbar^2\) is the determinant invariant of the selected
generator, rather than a Casimir eigenvalue of a representation.

On an interval where \(W\ne0\), the adapted coordinate satisfies 
\begin{equation}
 dz=\frac{dx}{W}=\frac{df}{P(f)},\qquad W\partial_x=\partial_z,
 \qquad \{f;z\}=PP''-\frac12(P')^2=\frac{2C}{\hbar^2},
 \label{eq:adapted}
\end{equation} 
where the derivatives of \(P\) are with respect to \(f\).
In \(z\), the flow introduced above is simply \(z\mapsto z+s\).
This straightening of a quadratic projective field is also used 
in AFF time reparametrization, \(d\tau=dt/(u+vt+wt^2)\)
\cite{AFF1976,InzunzaPlyushchay2019,AlPlyu}.
The discriminant classifies the flow in both settings. Here its
invariant fixes the charge-kernel energies, and \(z\) is spatial.

\subsection{Local reconstruction}
The Lie derivative of a weight-\(-1/2\) density along
\(W\partial_x\) is
\begin{equation}
 \mathfrak D_W=WD_x-\frac12W'.
 \label{eq:DW}
\end{equation} 
In the auxiliary basis used above, the same matrix \(\mathsf T_W\)
describes its componentwise action:
\[
 \mathfrak D_W
 \begin{pmatrix}\psi_1\\\psi_2\end{pmatrix}
 =\mathsf T_W
 \begin{pmatrix}\psi_1\\\psi_2\end{pmatrix}.
\]
Thus the projective flow of the ratio and the action on auxiliary
solutions are two realizations of the same selected symmetry. 
For \(\mathcal L=-\hbar^2(D_x^2+\mathcal U/2)\), covariance
under a stabilizing field, \(\Bol_3^{\mathcal U}W=0\), takes the
intrinsic form
\[
 \mathcal L\,\mathfrak D_W^{(-1/2)}
 =\mathfrak D_W^{(3/2)}\,\mathcal L,\qquad
 \mathfrak D_W^{(\rho)}=WD_x+\rho W'.
\]
The Lie derivatives act on the input and output density weights,
respectively. Since
\(\mathfrak D_W^{(3/2)}=\mathfrak D_W+2W'\), this is equivalent
to the first identity below in the fixed coordinate; direct
multiplication also gives the second: 
\begin{equation}
 [\mathcal L,\mathfrak D_W]=2W'\mathcal L,\qquad
 \hbar^2\mathfrak D_W^2+W^2\mathcal L=-C\mathbf1,
 \quad C=\frac{\hbar^2}{4}I[W,\mathcal U].
 \label{eq:generator-identities}
\end{equation} 
If \(\mathcal L\psi=0\), the first identity gives
\(\mathcal L\mathfrak D_W\psi=0\); the second then gives
\(\hbar^2\mathfrak D_W^2\psi=-C\psi\).
These are statements on the auxiliary solution space, not commutation
of \(\mathfrak D_W\) with \(\mathcal L\) on arbitrary functions.

In the reconstruction below, sums such as
\(\mathcal L\pm2\hbar\mathfrak D_W\) use the scalar coefficient
expressions in the fixed physical coordinate \(x\):
\(\mathcal L\) changes density weight, whereas \(\mathfrak D_W\)
preserves it. This coordinate identification yields the
constant-mass Schr\"odinger expressions. General reparametrizations
need not preserve that kinetic normalization.
 
\begin{proposition}[Local projective reconstruction]
\label{prop:reconstruction}
Fix \(\hbar>0\), an energy origin \(a\), and a physical coordinate
\(x\). Let \(\mathcal U(x)\) be a smooth real projective coefficient on an interval
and \(W\) a nontrivial real solution of
\(\Bol_3^{\mathcal U}W=0\). Define
\(C=\hbar^2I[W,\mathcal U]/4\) and choose \(\chi'=W\).
Then \(C\) is constant, and
\begin{align}
 2(H_--a)&=e^{-\chi/\hbar}
        (\mathcal L+2\hbar\mathfrak D_W)e^{\chi/\hbar},\nonumber\\
 2(H_+-a)&=e^{\chi/\hbar}
        (\mathcal L-2\hbar\mathfrak D_W)e^{-\chi/\hbar},
 \label{eq:H-reconstruction}\\
 2q_+&=-e^{-\chi/\hbar}\mathcal Le^{\chi/\hbar},\qquad
 2q_-=-e^{\chi/\hbar}\mathcal Le^{-\chi/\hbar}
 \label{eq:q-reconstruction}
\end{align}
satisfy \eqref{eq:superalgebra} as local differential expressions.
Where \(W\ne0\), they coincide with \eqref{eq:operators}.
Conversely, every member of that family determines
\((\mathcal U_C,W\partial_x)\).
\end{proposition}

\begin{proof}
The stabilizer equation gives \(I'=0\). For \(W\ne0\), solving
the invariant relation for \(-\hbar^2\mathcal U/2\) gives
\(\Delta-C/W^2\). Conjugation in \eqref{eq:H-reconstruction}
cancels the first derivative and yields
\(-\hbar^2D_x^2+W^2\mp2\hbar W'-\hbar^2\mathcal U/2\).
Equation~\eqref{eq:q-reconstruction} yields the charge and its formal
adjoint. Intertwining and closure follow from the factorization in
appendix~\ref{app:closure}. A nontrivial solution of the regular
third-order stabilizer equation cannot have \(W=W'=W''=0\) at a
point. Its zeros are isolated and of order at most two. The
reconstructed expressions in \(\mathcal U,W\) are smooth, so the
identities extend through these zeros by continuity. The converse
follows from \eqref{eq:invariant} and its derivative.
\end{proof}

The proposition is a local equivalence with the established family, 
including its smooth continuation through allowed zeros of \(W\).
Self-adjoint Hamiltonian domains and mutually adjoint charge domains
are additional physical data. 
At fixed \(\mathcal U\), selecting different stabilizers generally
changes the dressing, the difference between the partner potentials,
and the charge-kernel energies.
 These data therefore specify more than an auxiliary
second-order equation. They also replace the potentially singular
combination \(\Delta-C/W^2\) by the smooth coefficient
\(-\hbar^2\mathcal U/2\); section~\ref{sec:zeros} gives the resulting
regularity criterion for a prescribed \(W\).

The physical stationary equation has its own projective coefficient.
For a ratio \(f_{\pm,E}\) of two solutions of
\((H_\pm-E)\Psi=0\),
\begin{equation}
 \{f_{\pm,E};x\}
 =\mathcal U_C-\frac{2W^2}{\hbar^2}
       \mp\frac{4W'}\hbar+\frac{4(E-a)}{\hbar^2}.
 \label{eq:physical-connection}
\end{equation}
The reconstruction relates both physical spectral families to a
single auxiliary connection and its selected stabilizer.

Reversing \(W\) preserves \(\mathcal U_C,C\) and exchanges
the partners and charges; a simultaneous basis change and inverse
congruence of \(\mathsf M\) leaves the operators unchanged.

The relation to earlier differential parametrizations is explicit.
In a regular developing chart set \(\Gamma_f=f''/f'\), so that
\(\mathcal U=\Gamma_f'-\Gamma_f^2/2\). Then
\begin{equation}
 (D_x-\Gamma_f)D_x(D_x+\Gamma_f)W
   =\Bol_3^{\mathcal U}W=0.
 \label{eq:earlier-constraint}
\end{equation}
This is the third-order intertwining constraint of
refs.~\cite{Tanaka2003,BagchiTanaka2009}, expressed as preservation of
a projective connection. Their second-order Hamiltonians and charges
reduce to \eqref{eq:operators} after matching normalizations and the
energy origin.\footnote{In units \(\hbar=1\), the notation of
ref.~\cite{BagchiTanaka2009} is related to ours by
\(E(x)=\Gamma_f(x)\), \(R=-a\), and \(P_2^-=2q_+\).
The constraint is their Eq.~(12), and the operator expressions are
Eqs.~(39)--(40).} 
Likewise, the discriminant of \(P\) is
\(\operatorname{disc}P=4\varpi^2(M_{12}^2-M_{11}M_{22})
=-4C/\hbar^2\), the polynomial invariant already present in
that description. 
The reconstruction identifies the projective meaning of these
differential constraints and connects their invariant to the
Hamiltonian action, regularity and Bol--Lax comparison below.

\subsection{Projective classes and admissible singlets}
\label{sec:phases} 
On the auxiliary solution space, \eqref{eq:generator-identities}
gives \(\hbar^2\mathfrak D_W^2=-C\). By
\eqref{eq:charge-gauge}, multiplication by \(e^{-\chi/\hbar}\)
maps this space to \(\ker q_+\), while multiplication by
\(e^{\chi/\hbar}\) maps it to \(\ker q_-\).
On these charge kernels the Hamiltonians act as 
\begin{equation}
 (H_--a)e^{-\chi/\hbar}\psi
 =\hbar e^{-\chi/\hbar}\mathfrak D_W\psi,\qquad
 (H_+-a)e^{\chi/\hbar}\psi
 =-\hbar e^{\chi/\hbar}\mathfrak D_W\psi.
 \label{eq:kernel-action}
\end{equation} 
For the explicit square-root formulas below, work on a \(W>0\)
interval. For \(C=-\beta^2<0\), the auxiliary eigenvectors 
\(\psi_\sigma=\sqrt W e^{\sigma\beta z/\hbar}\),
\(\sigma=\pm1\), yield charge-kernel eigenfunctions at
\(a\pm\beta\). For \(C>0\) the two factorization energies
are complex conjugates. At \(C=0\), the generator is nilpotent:
with \(\Phi_0=\sqrt W e^{-\chi/\hbar}\) and
\(\Phi_1=z\Phi_0\),
\begin{equation}
 (H_--a)\Phi_0=0,\qquad (H_--a)\Phi_1=\hbar\Phi_0,
 \qquad q_+\Phi_{0,1}=0.
 \label{eq:confluent}
\end{equation}
Thus \(\Phi_1\) is an associated function, not a second
eigenfunction. Both cannot belong to the same self-adjoint domain:
their inner product would imply \(0=\hbar\|\Phi_0\|^2\).
The three alternatives are
\begin{center}
\begin{tabular}{lll}
\hline
Selected field & Real projective fixed points & Candidate singlet energies\\
\hline
Hyperbolic, \(C=-\beta^2<0\) & Two simple & \(a-\beta,\ a+\beta\)\\
Parabolic, \(C=0\) & One double & \(a\)\\
Elliptic, \(C>0\) & None & None\\
\hline
\end{tabular}
\end{center}
The energy alternatives were established in
ref.~\cite{KlishevichPlyushchay2001}; here they are the conjugacy
classes of the selected projective field.

To decide whether SUSY is unbroken, a charge zero mode must be
admissible in the specified spectral realization. For bound states
this requires square integrability. For periodic systems, Bloch
states obey \(\Psi(x+L)=e^{i\vartheta}\Psi(x)\), including the
corresponding derivative conditions, and are normalized on a cell.
A fixed \(\vartheta\) specifies a \emph{Bloch fiber}. 
In a regular finite-gap scalar system, the lowest band edge and
the endpoints of open finite gaps are nondegenerate spectral singlets.
Their states are periodic or antiperiodic, with Hill discriminant
\(+2\) or \(-2\), respectively; the discriminant is the trace of
the one-period translation matrix on the solution space. 
Away from closed-gap double points, an interior energy has two
independent Bloch solutions with distinct multipliers
\(e^{\pm i\vartheta}\). Their quasi-momenta \(\pm\vartheta/L\),
defined modulo \(2\pi/L\), label different fibers. At a
closed-gap double point, where a forbidden gap has collapsed,
the one-period translation matrix is \(+\mathbf1\) or
\(-\mathbf1\). Both solutions are then periodic or both are
antiperiodic: they share the same quasi-momentum modulo \(2\pi/L\)
and belong to the same fiber. These double points are distinct
from the nondegenerate endpoints of open gaps
\cite{CorreaJakubskyPlyushchay2008}. 
A spectral singlet is a SUSY singlet of a specified algebra only
if its supercharges annihilate it. Lack of full-line square
integrability does not by itself exclude a band-edge state from
this role in a Bloch fiber \cite{DunneFeinberg1998}. 

For a normalized energy eigenstate in the domains of the charge
products, with the cell inner product where appropriate, adjointness
and the closure give
\begin{equation}
 \|\mathcal Q_+\boldsymbol\Psi_E\|^2+
 \|\mathcal Q_-\boldsymbol\Psi_E\|^2=(E-a)^2+C.
 \label{eq:positivity}
\end{equation} 
Such an eigenstate is a SUSY singlet precisely when its energy
is a zero of \((E-a)^2+C\). If no admissible state at these
candidate energies is annihilated by the supercharges, this SUSY
is completely broken in the chosen realization.
Consequently \(C>0\) excludes singlets, whereas \(C\leq0\)
only permits them. 

The same positivity argument constrains the entire spectrum,
including its continuous part, if \(\mathcal H\) is self-adjoint
and mutually adjoint closed supercharges realize 
\eqref{eq:superalgebra} as an equality of closed quadratic forms.
Then \((\mathcal H-a)^2+C\mathbf1\geq0\), and for
\(C=-\beta^2<0\) the spectral theorem excludes
\((a-\beta,a+\beta)\) from \(\sigma(\mathcal H)\).
Formal differential closure alone does not establish this restriction
for an arbitrary self-adjoint realization.

The zero-gap free doublet shows the remaining freedom at fixed
\(C=-\beta^2\). For \(W=w_*>0\),
\begin{equation}
 E_{\rm th}=a+\tfrac12(w_*^2+\beta^2/w_*^2),\qquad
 E_{\rm th}-(a+\beta)=\frac{(w_*^2-\beta)^2}{2w_*^2},\qquad
 q_\pm1=\tfrac12(w_*^2-\beta^2/w_*^2).
 \label{eq:free-threshold}
\end{equation}
The constant band-bottom states are periodic SUSY singlets precisely
when \(w_*^2=\beta\); otherwise both polynomial roots lie below
the spectrum. In the former case there is one singlet in each graded
sector, so the zero-mode index, their count difference, vanishes
despite unbroken SUSY. Domain failure of a formal symmetry is a
separate issue \cite{CorreaDelOlmoPlyushchay2005}.

All three projective classes occur in the elementary family
\(W=\gamma/x\), \(\gamma>0\), on \(x>0\)
\cite{Plyushchay2017}:
\begin{equation}
 2(H_\pm-a)=-\hbar^2D_x^2+
 \frac{(\gamma\mp\hbar)^2-\hbar^2/4}{x^2}-\frac C{\gamma^2}x^2,
 \qquad z=\frac{x^2}{2\gamma},\qquad
 \mathcal U_C=-\frac{3}{2x^2}+\frac{2Cx^2}{\hbar^2\gamma^2}.
 \label{eq:inverse-example}
\end{equation}
The quadratic physical term is confining, absent or inverted as
\(C<0\), \(C=0\) or \(C>0\). These physical evolution
generators should not be confused with the \emph{auxiliary spatial}
generator: the confining example has hyperbolic auxiliary type and
the inverted one elliptic type. Endpoint conditions at the origin
remain part of the physical realization.

\subsection{Fixed points and regular superpotential zeros}
\label{sec:zeros}
The reconstruction is smooth in \(\mathcal U_C,W\), even if the
separate terms \(\Delta\) and \(C/W^2\) diverge. In a regular
developing chart,
\begin{equation}
 W=\frac{P(f)}{f'},\qquad f'\ne0.
 \label{eq:fixed-points}
\end{equation} 
Because \(f'\ne0\), a zero of \(W\) corresponds to a fixed
point of \(P(f)\partial_f\) with the same multiplicity.
 The quadratic
\(P\) permits simple hyperbolic zeros, double parabolic zeros and
no elliptic zeros; the reciprocal chart covers projective infinity.

For a prescribed smooth \(W\), the geometric criterion translates
into explicit cancellation conditions. Put
\begin{equation}
 U_C=\frac{\mathcal N_C}{W^2},\qquad
 \mathcal N_C=\frac{\hbar^2}{2}WW''-\frac{\hbar^2}{4}W'^2-C,
 \qquad \mathcal N_C'=\frac{\hbar^2}{2}WW'''.
 \label{eq:numerator}
\end{equation}
\begin{proposition}[Regular superpotential zeros]
\label{prop:regular-zeros}
Let \(W\) be smooth and nonzero in a punctured neighborhood of
\(x_0\). At a simple zero, the necessary and sufficient condition
for smooth continuation of \(U_C\), and its resulting value, are
\begin{equation}
 C=-\frac{\hbar^2W'(x_0)^2}{4};\qquad
 U_C(x_0)=\frac{\hbar^2W'''(x_0)}{4W'(x_0)}.
 \label{eq:simple-zero}
\end{equation}
At a double zero, the conditions and limiting value are
\begin{equation}
 C=0,\qquad W'''(x_0)=0;\qquad
 U_0(x_0)=\frac{\hbar^2W''''(x_0)}{4W''(x_0)}.
 \label{eq:double-zero}
\end{equation}
Smooth continuation is impossible at a zero of higher order or an
infinitely flat zero.
\end{proposition}
\begin{proof}
At a simple zero, cancellation requires
\(\mathcal N_C(x_0)=0\); then also \(\mathcal N_C'(x_0)=0\).
Numerator and denominator have smooth quadratic factors, and their
quotient gives \eqref{eq:simple-zero}. At a double zero, boundedness
first requires \(C=0\). Write \(W=t^2v(t)\), \(t=x-x_0\),
and use the signed amplitude \(R=t\sqrt{|v|}\).
Then \(U_0=\hbar^2R''/R\) is smooth exactly when
\(v'(0)=0\), equivalently \(W'''(x_0)=0\), yielding 
\eqref{eq:double-zero}. If \(U_C\) extended smoothly through
\(x_0\), then \(\mathcal U_C=-2U_C/\hbar^2\) would be smooth,
and the stabilizer equation on the punctured neighborhood would
extend by continuity to a regular third-order equation across \(x_0\).
A higher-order or infinitely flat zero would then give
\(W(x_0)=W'(x_0)=W''(x_0)=0\), forcing \(W\equiv0\) locally
by uniqueness and contradicting the hypothesis. 
\end{proof}
Thus a single \(C\) permits several simple zeros only if their
absolute slopes agree, and permits double zeros only in the
parabolic class. These conditions test regularity directly in the
original SUSY parametrization. At an allowed zero, the divergence
of \(z'=1/W\) marks the failure of the straightening coordinate;
the reconstructed operators and a suitable projective ratio chart
remain regular.

\subsection{The distinction from higher orders}
\label{sec:higher}
Projective covariance is not exclusive to second order. For example,
the undeformed charges \((\hbar D_x+W)^n\), \(n>1\), require
\(W'''=0\), the flat stabilizer equation
\cite{KlishevichPlyushchay2001}. Gauging their kernel to
\(\mathcal V_n=\operatorname{span}\{1,x,\ldots,x^{n-1}\}\) gives
\begin{equation}
 e^{\chi/\hbar}(H_{n,-}-a)e^{-\chi/\hbar}
 =-\tfrac12\hbar^2D_x^2+
 \hbar\bigl(WD_x-\tfrac{n-1}{2}W'\bigr),\qquad
 2(H_{n,-}-a)=-\hbar^2D_x^2+W^2-n\hbar W'.
 \label{eq:kernel-n}
\end{equation} 
Here \(\chi'=W\), and conjugation by \(e^{\chi/\hbar}\)
turns the charge into \(\hbar^nD_x^n\), whose kernel is
\(\mathcal V_n\). For quadratic \(W\), the first-order term is an 
\(\mathfrak{sl}(2,\mathbb R)\) generator on \(\mathcal V_n\),
the familiar mechanism of quasi-exact solvability
\cite{Turbiner1988,AoyamaSatoTanaka2001}. At \(n=2\),
\(D_x^2\) vanishes on this space; the Hamiltonian action reduces
to one generator with a quadratic invariant. For \(n>2\) the
kinetic term survives. In the broader polynomial-kernel construction,
an additional compatibility condition is required for \(n\geq3\),
together with the stabilizer equation \eqref{eq:earlier-constraint} 
\cite{Tanaka2003}. The polynomial example illustrates the
second-order distinction; the broader compatibility conditions are
established in the cited literature. The reconstruction above
identifies the geometric data of the second-order family.
A general higher-order theory requires more than one
selected-generator invariant. 

\section{Geometric origin and transport of the correction}
\label{sec:curved}
\subsection{The connection selected by a curved measure}
The affine data of the fictitious factorization have a concrete
realization in dimensional reduction. This realizes the scalar-operator
aspect of the curved-space relation suggested in
ref.~\cite{Plyushchay2017}; earlier connections between nonlinear
SUSY and curved spaces appear in
refs.~\cite{KlishevichPlyushchayRiemann,AnabalonPlyushchay2003}.
On a \(W>0\) interval take
\begin{equation}
 ds^2=dx^2+F^2d\varphi^2,\qquad F=W/\Omega,\qquad
 \chi'=W,\qquad \Omega>0.
 \label{eq:metric}
\end{equation}
The cyclic coordinate need not be periodic. With
\(\Delta_g=F^{-1}\partial_x(F\partial_x)+F^{-2}\partial_\varphi^2\)
the scalar Laplace--Beltrami operator, one has
\(\nabla_\mu\nabla_\nu\chi=W'g_{\mu\nu}\),
\(\Delta_g\chi=2W'\), and \(|d\chi|^2=W^2\).
Consequently the zero- and two-form sectors of Witten's operator
\(\tfrac12\{d_\chi,d_\chi^\dagger\}\),
\(d_\chi=\hbar d+d\chi\wedge\), have scalar expressions
\(\tfrac12[-\hbar^2\Delta_g+W^2\mp2\hbar W']\)
\cite{Witten1982}.

Fixing \(-i\hbar\partial_\varphi=\hbar k_\varphi\) and using
\(\mathsf T_Fu=\sqrt F\,u\) to pass from \(L^2(Fdx)\)
to \(L^2(dx)\) gives
\begin{equation}
 \mathsf T_F[-\hbar^2(D_x^2+F'D_x/F)]\mathsf T_F^{-1}
 =-\hbar^2D_x^2+\Delta(W),\qquad C=-(\hbar\Omega k_\varphi)^2.
 \label{eq:reduction}
\end{equation}
The reduced scalar operators, shifted by \(a\), are precisely
\(H_\mp\): the measure produces \(\Delta\), and cyclic momentum
produces \(-C/W^2\). More specifically,
\begin{equation}
 \Gamma_\zeta=-F'/F,\qquad
 ds^2=W^2(dz^2+\Omega^{-2}d\varphi^2),\qquad
 \Delta=-\frac{\hbar^2}{2}\{z;x\},\qquad dz=dx/W.
 \label{eq:measure-connection}
\end{equation}

For the specified metric and scalar Witten operator, the same
coordinate that makes the metric conformally flat selects the affine
connection in the SUSY factorization. The reduced measure therefore
produces the required Schwarzian correction without an additional
curvature coupling. This realizes the SUSY factorization data
geometrically; it does not assert uniqueness of quantization from
the metric alone.

This Riemannian reduction gives \(C\leq0\). A Lorentzian cyclic
direction reverses the sign at the level of scalar differential
expressions, but needs a separate parent-theory analysis. The full
parent supercharges and their grading, including the Witten and
Clifford--Pauli constructions, will be developed in a subsequent
paper. The present result concerns the scalar partners and their
measure; regularity of the reduced operators need not imply that
the parent metric extends through a zero of \(W\).

\subsection{What a variable mass does and does not fix}
\label{sec:PDM} 
A position-dependent mass fixes the kinetic metric, but leaves a
choice of factor ordering \cite{BravoPlyushchay2016}.
Let \(\mu(x)>0\) be the mass, \(u(x)\) the potential, and
\(b=(2\mu)^{-1/2}\). Using the factorization function \(\zeta\),
quantize 
\(p^2/(2\mu)+u\) as \(H_{b,\zeta}=bK_\zeta b+u\).
The coordinate \(dy=dx/b\) and the unitary amplitude
\((\mathsf T\psi)(y)=\sqrt{b(x(y))}\psi(x(y))\) give
\begin{equation}
 \mathsf TH_{b,\zeta}\mathsf T^{-1}
 =-\hbar^2D_y^2+\hbar^2(\omega_\Phi^2-D_y\omega_\Phi)+u(x(y)),
 \quad \Phi=\sqrt b\,\zeta,\quad \omega_\Phi=D_y\log\Phi.
 \label{eq:PDM-flat}
\end{equation}
The surviving coefficient is the difference of two affine connections:
\begin{equation}
 (\Gamma_\zeta-\Gamma_\mu)\,dx=2\omega_\Phi\,dy,
 \qquad \Gamma_\mu=\tfrac12(\log\mu)'=-b'/b.
 \label{eq:relative}
\end{equation}
Their difference is a one-form because their inhomogeneous
transformation terms cancel. Flattening the kinetic metric therefore
does not in general remove the factorization data: the first-order
factor \(\hbar(D_y+\omega_\Phi)\) becomes a plain derivative
only when the two connections coincide. In the SUSY reduction above,
the measure fixes the factorization connection; a mass profile alone
would not determine it. This is the role of variable mass in the
geometric interpretation of the fictitious transformation.

\subsection{Metamorphosis in the adapted projective coordinate}
\label{sec:stationary-metamorphosis}
The same coordinate \(z'=1/W\) implements coupling-constant
metamorphosis, which exchanges an energy parameter with a coupling
\cite{Plyushchay2017,HietarintaEtAl1984,KalninsMillerPost2010}. 
On a \(W>0\) interval, set \(\varepsilon=E-a\) and
\(\Psi_\pm=\sqrt W\,\widetilde\Psi_\pm(z)\). Multiplying 
the transformed stationary equation by \(2W^2\) gives
\begin{gather}
 [-\hbar^2D_z^2+\mathcal W_\varepsilon^2
       \pm\hbar\partial_z\mathcal W_\varepsilon]
       \widetilde\Psi_\pm=\mathcal E_\varepsilon\widetilde\Psi_\pm,
 \nonumber\\
 \mathcal W_\varepsilon=W^2-\varepsilon,\qquad
 \mathcal E_\varepsilon=\varepsilon^2+C.
 \label{eq:metamorphosis}
\end{gather}
The Schwarzian term cancels the derivative contribution of the
amplitude. The identity \(\mathfrak D_W(\sqrt W\,\phi(z))=
\sqrt W\,\phi_z\) reduces the quadratic charges on stationary
solutions to
\begin{equation}
 q_+\Psi_-=\sqrt W(\hbar D_z+\mathcal W_\varepsilon)
       \widetilde\Psi_-,\qquad
 q_-\Psi_+=\sqrt W(-\hbar D_z+\mathcal W_\varepsilon)
       \widetilde\Psi_+.
 \label{eq:stationary-charges}
\end{equation} 
Here \(W\) is evaluated at \(x=x(z)\). For each fixed
\(\varepsilon\), these are first-order SUSY partners:
the original energy parameter enters their superpotential, while
the original coupling \(C\) enters the transformed energy
\(\mathcal E_\varepsilon\).
 For \(W=\gamma/x\), one obtains
\(\mathcal W_\varepsilon=\gamma/(2z)-\varepsilon\), the Coulomb
superpotential, with energy relative to its asymptotic threshold
\(\mathcal E_\varepsilon-\varepsilon^2=C\).

This order reduction uses the stationary equation. It is not an
operator similarity or an exchange of two coexisting physical SUSYs.
The amplitude here is \(\sqrt W\), whereas a unitary coordinate
map from \(L^2(dz)\) to \(L^2(dx)\) uses \(W^{-1/2}\).
Charge-annihilation relations are transported, but normalizability
and boundary conditions must be checked anew. In particular, the
map is restricted to intervals on which \(W\) has fixed sign;
a zero of \(W\) can become an infinite endpoint in \(z\).

\section{The Lam\'e pair: projective geometry and integrability}
\label{sec:integrable} 
The local reconstruction does not require a physical Lax--Novikov
integral. The one-gap Lam\'e pair tests how such an integral relates
to the SUSY auxiliary geometry: it has first- and second-order SUSYs
and a third-order Lax--Novikov integral, yet its auxiliary third
Bol operator is not that integral. 
We use this established system
\cite{CorreaEtAl2008,CorreaJakubskyPlyushchay2008,PlyushchayArancibiaNieto2011}
to exhibit the geometric distinction and its reflectionless limit.

\subsection{Reconstruction and half-period inversion}
\label{sec:lame}
We take explicitly \(m=1,l=0\) in the associated Lam\'e family:
the intertwiner orders \(m-l\) and \(m+l+1\) are then one and
two. Set \(0<k<1\), \(r=k'=\sqrt{1-k^2}\), 
\(K=K(k)\), the complete elliptic integral of the first kind,
and \(\Sigma=1+r^2\). The Jacobi functions 
\(\operatorname{sn}x,\operatorname{cn}x,\operatorname{dn}x\)
all have modulus \(k\), and the inverse length scale is set to one.
The pair is
\begin{equation}
 L_-=-D_x^2+\Sigma-2\operatorname{dn}^2x,\qquad
 L_+=-D_x^2+\Sigma-\frac{2r^2}{\operatorname{dn}^2x},\qquad
 H_\pm=\frac{\hbar^2}{2}L_\pm.
 \label{eq:lame-H}
\end{equation}
The identity \(\operatorname{dn}(x+K)=r/\operatorname{dn}x\)
interchanges the functional factors
\(\operatorname{dn}^2x\) and \(r^2/\operatorname{dn}^2x\).
For this pair it translates one potential into the other.
Its projective reconstruction data are
\begin{align}
 W&=-\frac{\hbar}{2}D_x\log\operatorname{dn}x
   =\frac{\hbar k^2\operatorname{sn}x\operatorname{cn}x}
          {2\operatorname{dn}x},\qquad
 a=\frac{\hbar^2\Sigma}{4},\qquad C=-\frac{\hbar^4k^4}{16},\nonumber\\
 \mathcal U_C&=\frac32\left(\operatorname{dn}^2x+
                       \frac{r^2}{\operatorname{dn}^2x}\right)
                       -\frac{\Sigma}{2}.
 \label{eq:lame-data}
\end{align}
An auxiliary basis and developing coordinate are
\begin{equation}
 \psi_1=\frac{\sqrt r\,\operatorname{sn}x}{\sqrt{\operatorname{dn}x}},
 \qquad \psi_2=\frac{\operatorname{cn}x}{\sqrt{\operatorname{dn}x}},
 \qquad f=\sqrt r\,\frac{\operatorname{sn}x}{\operatorname{cn}x},
 \qquad W\partial_x=\frac{\hbar k^2}{2}f\partial_f.
 \label{eq:lame-projective}
\end{equation}
Thus \(\{f;x\}=\mathcal U_C\) and the selected generator is
hyperbolic. The zeros of \(W\) have slopes
\(\pm\hbar k^2/2\), satisfying \eqref{eq:simple-zero}; the
auxiliary connection and SUSY operators remain smooth there.
The zero and pole of the ratio chart represent the two projective
fixed points, with a reciprocal chart at the pole.

The half-period translation has the projective action
\begin{equation}
 f(x+K)=-\frac1{f(x)},\qquad
 \mathsf S=\begin{pmatrix}0&-1\\1&0\end{pmatrix},\qquad
 \mathsf S^2=-\mathbf1.
 \label{eq:lame-inversion}
\end{equation}
It preserves \(\mathcal U_C\), reverses \(W\), and exchanges
the physical partners. This discrete elliptic transformation has
order two in \(\mathrm{PSL}(2,\mathbb R)\); it interchanges the
fixed points of the selected hyperbolic field. Its integer matrix
acts on the solution ratio and does not by itself implement a
modular transformation of the elliptic period ratio. The example
therefore separates the connection, its selected infinitesimal
generator, and a discrete symmetry exchanging the partners.

\subsection{Auxiliary symmetry versus physical integrability}
\label{sec:GD}
The physical intertwiners needed for the comparison are
\begin{equation}
 A=D_x-(\log\operatorname{dn}x)',\qquad
 Y=D_x^2-(\log\operatorname{dn}x)'D_x+\operatorname{dn}^2x,
 \qquad q_+=\frac{\hbar^2}{2}Y.
 \label{eq:lame-charges}
\end{equation}
They obey \(AL_-=L_+A\), \(YL_-=L_+Y\), and
\begin{equation}
 A^\dagger A=L_-,\qquad
 Y^\dagger Y=(L_--r^2)(L_--1).
 \label{eq:lame-closures}
\end{equation}
The plus-sector identities follow by translation. Placing these
maps in off-diagonal matrix entries gives the two SUSY orders;
their mixed products give even integrals of differential order three.

To compare the auxiliary Bol operator with the physical Lax integral,
recall the two distinct third-order operators in KdV. The GD
operator defines the second Poisson bracket on the space of
potentials, while the Lax evolution operator acts on Schr\"odinger
wavefunctions. For \(L=-D_x^2+v\), take these operators as
\begin{equation}
 \mathcal J_v=D_x^3-4vD_x-2v',\qquad
 B_0=D_x^3-\frac32vD_x-\frac34v'.
 \label{eq:GD-Lax-operators}
\end{equation}
Their different relative coefficients reflect these distinct roles;
see section~4.1 of ref.~\cite{Zuber1993}. In particular,
\(\mathcal J_v\) maps functional gradients to variations of \(v\).
In this normalization \(4[B_0,L]=\mathcal J_vv=v'''-6vv'\),
with multiplication understood on the right. The identification
\(\mathcal J_v=\Bol_3^{-2v}\) gives the geometric meaning:
\(\mathcal J_v\xi\) is the coadjoint variation of the projective
coefficient \(\mathcal U=-2v\) in \eqref{eq:variation}.
The third derivative is the contribution of the Gelfand--Fuchs
cocycle to the Virasoro coadjoint action \cite{Ovsienko2006}.

For a general third Bol operator, coefficient comparison gives
\begin{equation}
 [\Bol_3^{\mathcal U},-D_x^2+v]=0
 \quad\Longleftrightarrow\quad
 \mathcal U=-\frac34v+c_0,\qquad v'''-6vv'+8c_0v'=0,
 \label{eq:stationary-KdV}
\end{equation}
with constant \(c_0\).
The commuting operator is then
\(B=B_0+2c_0D_x\).

The second equation is the stationary Korteweg--de Vries condition.
Commutation imposes compatibility with the physical potential, in
addition to the auxiliary stabilizer equation
\(\Bol_3^{\mathcal U_C}W=0\).

The Bol--Lax comparison here uses differential expressions in the
fixed physical coordinate \(x\). Geometrically, \(\Bol_3\) maps
density weights \(-1\) to \(2\), while the Lax operator acts on
wavefunctions. Equality of their coefficient expressions uses the
coordinate identification, as in section~\ref{sec:reconstruction};
their respective density actions must be distinguished under
reparametrization. In particular, the relation
\(\mathcal U=-3v/4+c_0\) specifies the commuting expression in
\(x\), rather than the projective coefficient of the physical
Schr\"odinger equation.

For \(v=\Sigma-2\operatorname{dn}^2x\), it gives
\begin{align}
 \mathcal U_{\mathrm L,-}&=\tfrac32\operatorname{dn}^2x-\tfrac12\Sigma,
 \qquad B_-:=\Bol_3^{\mathcal U_{\mathrm L,-}}=-A^\dagger Y,
 \qquad \mathscr P_-=i\hbar^3B_-,\nonumber\\
 [\mathscr P_-,L_-]&=0,\qquad
 \mathscr P_-^2=\hbar^6L_-(L_--r^2)(L_--1),\qquad c_0=\Sigma/4.
 \label{eq:lame-Lax}
\end{align}
Here \(\mathscr P_-\) is the formally Hermitian Lax--Novikov
integral; its partner is obtained by translation.
The decisive comparison is
\begin{equation}
 \mathcal U_C-\mathcal U_{\mathrm L,-}
 =\frac{3r^2}{2\operatorname{dn}^2x}.
 \label{eq:lame-distinction}
\end{equation}
At finite modulus this difference is nonconstant. Hence
\(\Bol_3^{\mathcal U_C}\) cannot be the commuting third-order
integral: the auxiliary projective symmetry of second-order SUSY
and the physical Lax symmetry use different coefficients.

The physical stationary equation determines an energy-dependent
projective coefficient, \(\mathcal U_\lambda=-2(v-\lambda)\), 
with \(\lambda=2E/\hbar^2\). Products of its solutions, and
their linear combinations, satisfy the Gelfand--Dikii equation.
Writing such a combination as \(\mathcal R(x,\lambda)\), one has 
\cite{Brezhnev2008Integrability,GelfandDikii1975},
\begin{equation}
 \mathcal R'''-4(v-\lambda)\mathcal R'-2v'\mathcal R=0,
 \qquad
 \nu^2=-\tfrac12\mathcal R\mathcal R''+
       \tfrac14(\mathcal R')^2+(v-\lambda)\mathcal R^2.
 \label{eq:GD}
\end{equation}

The second expression in \eqref{eq:GD} is an \(x\)-independent
first integral, denoted by \(\nu^2\); it need not be positive.
The differential operator in the first equation is
\(\mathcal J_\lambda=\mathcal J_v+4\lambda D_x
=\Bol_3^{\mathcal U_\lambda}\). 
Thus \(\mathcal R\partial_x\)
preserves the physical projective connection. For the Lam\'e potential,
\begin{equation}
 \mathcal R=\lambda-\operatorname{dn}^2x,\qquad
 \nu^2=-\lambda(\lambda-r^2)(\lambda-1).
 \label{eq:lame-GD}
\end{equation}
It encodes the spectral curve of \(\mathscr P_-\), whereas
\(\mathcal U_C\) reconstructs the SUSY pair at fixed \(C\).

The solution-product description also supplies a first-order
symmetry of the spectral equation; see section~3 of
ref.~\cite{Brezhnev2008Integrability}. In our notation its action is
\(\mathfrak D_{\mathcal R}=\mathcal R D_x-\mathcal R'/2\),
the Lie derivative on weight-\(-1/2\) densities, as in
\eqref{eq:DW}. Substitution of \eqref{eq:lame-GD} gives

\begin{equation}
 B_-=-\mathfrak D_{\mathcal R}-D_x\circ(L_--\lambda),\qquad 
 (L_--\lambda)\psi=0\ \Longrightarrow\
 B_-\psi=-\mathfrak D_{\mathcal R}\psi.
 \label{eq:Lax-projective-action}
\end{equation}

This established finite-gap mechanism permits a precise comparison
with \eqref{eq:kernel-action}: the physical Lax action is represented
by the stabilizer of \(\mathcal U_\lambda\), while the Hamiltonian
action on the charge kernel is represented by the stabilizer of
\(\mathcal U_C\). The coefficients serve three distinct purposes:
\(\mathcal U_C\), together with \(W\), reconstructs the SUSY pair;
\(\mathcal U_\lambda\) governs physical solution products; and
\(\mathcal U_{\mathrm L,-}\) expresses the commuting Lax operator
in Bol form in \(x\). The comparison separates the auxiliary
geometry selected by SUSY from the established spectral construction.

The band-edge data illustrate the separate role of admissibility.
The spectrum of \(L_-\) is \([0,r^2]\cup[1,\infty)\).
The periodic singlet \(\operatorname{dn}x\) at zero is killed
by \(A\); the antiperiodic singlets \(\operatorname{cn}x\) and
\(\operatorname{sn}x\) at \(r^2\) and one are killed by \(Y\).
Thus each SUSY order has singlets in the full Bloch spectrum, but
none is common to both orders at finite modulus. Restricting to the
periodic or antiperiodic fiber selects which order retains zero modes,
without changing \(\mathcal U_C\) or \(C\).

\subsection{Why the reflectionless limit is special}
\label{sec:soliton-limit}
At fixed \(x\), the limit \(k\to1\) gives
\(r\to0\), \(K\to\infty\), and
\(\operatorname{dn}x\to\operatorname{sech}x\). The minus
subsystem becomes a reflectionless P\"oschl--Teller potential and 
the plus subsystem becomes free with a constant threshold shift, \(L_+\to-D_x^2+1\). In this limit 
\begin{equation}
 L_-=-D_x^2+1-2\operatorname{sech}^2x,\qquad
 \mathcal U_C=\mathcal U_{\mathrm L,-}
 =\tfrac12(3\operatorname{sech}^2x-1),\qquad
 \mathscr P_-=i\hbar^3\Bol_3^{\mathcal U_C}.
 \label{eq:Lax-Bol}
\end{equation}
The ratio in \eqref{eq:lame-projective} has a degenerate normalization;
rescaling it by \(r^{-1/2}\), which leaves its Schwarzian unchanged,
gives the regular limit \(f=\sinh x\). Thus the coincidence in
\eqref{eq:Lax-Bol} is a limit of the projective coefficients,
not a consequence of a singular choice of ratio chart.

The coincidence is an equality of differential expressions in
\(x\), with the density interpretation specified above. The physical
GD operator \(\mathcal J_\lambda\) remains distinct in this limit.

The surviving Darboux maps are \(A=D_x+\tanh x\) and
\(Y=D_xA\). Their mixed product still yields
\(B_-=-A^\dagger Y\), while the spectral relation degenerates to
\begin{equation}
 \mathscr P_-^2=\hbar^6L_-^2(L_--1).
 \label{eq:soliton-spectral}
\end{equation}
The auxiliary pair
\(\psi_2=(\cosh x)^{-1/2}\),
\(\psi_1=\sinh x/\sqrt{\cosh x}\) now has a symmetric square
equal to the Lax kernel. An equivalent basis is
\(\operatorname{sech}x,\tanh x,\cosh x\): respectively a bound
state at zero, a bounded threshold state at \(L_-=1\), and an
associated function satisfying
\(L_-\cosh x=-2\operatorname{sech}x\)
\cite{CorreaPlyushchay2007}.
The bound state is annihilated by both SUSY orders, and the threshold
states by the second-order maps. These different spectral roles
explain the double and simple roots in \eqref{eq:soliton-spectral}.
The auxiliary Bol operator has become a physical integral precisely
because the extra term in \eqref{eq:lame-distinction} has vanished.

\section{Discussion and outlook}
\label{sec:discussion}

With the physical coordinate, unit-mass kinetic normalization,
\(\hbar\), and energy origin fixed, a projective connection and a
selected stabilizing field determine all local differential operators
of the established second-order SUSY family. The connection fixes
\(\ker_{\rm loc}\mathcal L_C\); the field fixes \(\chi'=W\),
and hence the dressed local charge kernels
\(\ker_{\rm loc}q_\pm=e^{\mp\chi/\hbar}
\ker_{\rm loc}\mathcal L_C\), the partner splitting, and the
Hamiltonian action on those kernels.  
The selected field's quadratic invariant fixes the central parameter
of the nonlinear superalgebra. 
 This correspondence organizes the known intertwining
constraints into geometric data with explicit spectral and
regularity consequences.

The conjugacy invariant determines candidate singlet energies;
compatible domains and boundary conditions decide their admissibility.
For periodic systems this includes the Bloch fibers and band-edge
singlets. Proposition~\ref{prop:regular-zeros} gives a direct test
of smoothness in the original superpotential parametrization:
simple zeros must have a common absolute slope fixed by \(C<0\),
while double zeros require \(C=0\) and a vanishing third derivative.
The projective reconstruction remains regular when the adapted
coordinate fails at such a fixed point.

The affine factorization gives the classically invisible insertion
a quantum interpretation: its covariant factors generate the
Schwarzian correction required by SUSY. The curved scalar reduction
realizes these data through the measure, and the variable-mass
analysis shows why flattening the kinetic metric does not generally
erase them. These constructions identify the geometric content of
the factorization prescription while keeping its additional quantum
data explicit.

The Lam\'e pair relates the SUSY reconstruction to the established
Gelfand--Dikii--Virasoro structure. Its half-period translation acts
as a projective inversion exchanging the partners.  
The first-order representation of the
Lax action on physical solutions, \eqref{eq:Lax-projective-action}, 
uses a spectral stabilizer and can be compared directly with the
auxiliary stabilizer acting on the charge kernel. The stationary
KdV condition \eqref{eq:stationary-KdV} specifies when a Bol
expression in the physical coordinate is also a commuting integral.
In the reflectionless limit, the SUSY auxiliary Bol expression
coincides with the Lax integral up to normalization, and the
symmetric square of the auxiliary second-order solution space
is the Lax kernel.
 The physical GD
spectral connection remains distinct. The contribution of this
comparison is to identify the additional SUSY auxiliary geometry
and its compatibility with physical integrability.

Several extensions remain. The local reconstruction should be combined
with global monodromy, the change of an auxiliary solution basis
after continuation around a closed loop, and with compatible physical
domains. This is particularly relevant when regular zeros obstruct a
single adapted coordinate or when metamorphosis changes endpoint
conditions. The complete parent-supercharge construction, including
its grading and the Witten and Clifford--Pauli realizations, will be
developed in a subsequent paper. Higher-order SUSY requires additional
data beyond a selected quadratic invariant; its relation to invariant
polynomial flags and exceptional orthogonal polynomials
\cite{GomezUllateKamranMilson2011,InzunzaPlyushchay2019Hidden}
offers a separate extension of the present framework.

\section*{Acknowledgments}
The work was partially supported by the FONDECYT Project 1242046.

\appendix
\section{Operator checks}
\label{app:proofs}
\subsection{The correction prescription}
\label{app:correction}
Keep \(H_\pm=H_\pm^{(0)}+\delta V/2\). A general second-order
intertwiner has constant leading coefficient by the \(D_x^3\)
condition. Fix it to \(\hbar^2/2\) and write
\(q_+=q_+^{(0)}+A(x)D_x+B(x)\). The next two coefficients give
\begin{equation}
 A'=0,\qquad B'=-\tfrac12\delta V'+2AW'/\hbar,
 \qquad A=A_0,\quad B=-\tfrac12\delta V+2A_0W/\hbar+B_0.
 \label{eq:correction-constants}
\end{equation}
Retaining the specified coefficient \(\hbar W\) of \(D_x\)
sets \(A_0=0\). Absorbing \(2B_0\) into \(\delta V\)
shifts the common energy origin by \(B_0\); fixing \(a\) sets
\(B_0=0\). The scalar remainder is then 
\(W\delta V'+2W'\delta V-\hbar^2W'''/2=0\).
Multiplication by \(W\) and integration give
\[
 (W^2\delta V)'=\frac{\hbar^2}{2}WW''',\qquad
 W^2\delta V=\frac{\hbar^2}{2}WW''
             -\frac{\hbar^2}{4}W'^2+\kappa.
\]
Thus \(\delta V=\Delta+\kappa/W^2\); the integration constant
is absorbed by replacing \(C\) with \(C-\kappa\).
This explains the scope of the prescription 
\cite{Plyushchay2017}: additional quantum shifts of \(W\) or
\(a\) would be different choices of quantum data.

\subsection{Intertwining and quadratic closure}
\label{app:closure}
For \(C=-\beta^2<0\), introduce
\begin{equation}
 B_1=\hbar D_x+W-\frac{\hbar W'}{2W}+\frac\beta W,
 \qquad B_2=\hbar D_x+W+\frac{\hbar W'}{2W}-\frac\beta W.
\end{equation}
Direct multiplication yields
\begin{align}
 B_2B_1&=2q_+,\qquad
 B_1^\dagger B_1=2(H_--a+\beta),\qquad
 B_2B_2^\dagger=2(H_+-a-\beta),\nonumber\\
 B_1B_1^\dagger-B_2^\dagger B_2&=4\beta.
 \label{eq:B-identities}
\end{align}
The common intermediate operator
\(H_{\rm int}=B_1B_1^\dagger/2+a-\beta
=B_2^\dagger B_2/2+a+\beta\) proves intertwining by composition.
It also gives
\begin{equation}
 q_-q_+=(H_--a)^2-\beta^2,\qquad
 q_+q_-=(H_+-a)^2-\beta^2.
\end{equation}
For fixed nonvanishing \(W\), both intertwining and closure
remainders are polynomial in \(C\). Their vanishing for every
negative \(C\) therefore proves the identities for all real
\(C\). This extends differential identities, not physical domains;
it does not require treating complex factors as real adjoints.

\end{document}